\documentclass[english]{article}
\usepackage[T1]{fontenc}
\usepackage[utf8]{luainputenc}
\usepackage{amsthm}
\usepackage{amsmath}
\usepackage{amssymb}

\makeatletter
\numberwithin{equation}{section}
\numberwithin{figure}{section}
\theoremstyle{plain}
\newtheorem{thm}{\protect\theoremname}
\theoremstyle{definition}
\newtheorem{defn}[thm]{\protect\definitionname}
\theoremstyle{remark}
\newtheorem{claim}[thm]{\protect\claimname}
\ifx\proof\undefined\
\newenvironment{proof}[1][\protect\proofname]{\par
\normalfont\topsep6\p@\@plus6\p@\relax
\trivlist
\itemindent\parindent
\item[\hskip\labelsep
\scshape
#1]\ignorespaces
}{%
\endtrivlist\@endpefalse
}
\fi
\theoremstyle{plain}
\newtheorem{lem}[thm]{\protect\lemmaname}
\theoremstyle{plain}
\newtheorem{fact}[thm]{\protect\factname}
\newcounter{casectr}
\newenvironment{caseenv}
{\begin{list}{{\itshape\ \protect\casename} \arabic{casectr}.}{%
\setlength{\leftmargin}{\labelwidth}
\addtolength{\leftmargin}{\parskip}
\setlength{\itemindent}{\listparindent}
\setlength{\itemsep}{\medskipamount}
\setlength{\topsep}{\itemsep}}
\setcounter{casectr}{0}
\usecounter{casectr}}
{\end{list}}

\usepackage{complexity}
\usepackage{pgfplots}
\usepackage{pgf}
\usetikzlibrary{arrows}

\usepackage[hyperfootnotes]{hyperref}

\makeatother

\usepackage{babel}
\providecommand{\casename}{Case}
\providecommand{\claimname}{Claim}
\providecommand{\definitionname}{Definition}
\providecommand{\factname}{Fact}
\providecommand{\lemmaname}{Lemma}
\providecommand{\proofname}{Proof}
\providecommand{\theoremname}{Theorem}

\begin{document}
\author{Thatchaphol Saranurak\thanks{Universität des Saarlandes, \href{mailto:thatchaphol.s@stud.uni-saarland.de}{thatchaphol.s@stud.uni-saarland.de}} 
\and 
Gorav Jindal\thanks{Universität des Saarlandes, \href{mailto:s9gojind@stud.uni-saarland.de}{s9gojind@stud.uni-saarland.de}}}

\title{Subtraction makes computing integers faster}
\maketitle
\begin{abstract}
We show some facts regarding the question whether, for any number
$n$, the length of the shortest Addition Multiplications Chain (AMC)
computing $n$ is polynomial in the length of the shortest division-free
Straight Line Program (SLP) that computes $n$.

If the answer to this question is {}``yes'', then we can show a
stronger upper bound for $\mathrm{PosSLP}$, the important problem
which essentially captures the notion of efficient computation over
the reals. If the answer is {}``no'', then this would demonstrate
how subtraction helps generating integers super-polynomially faster,
given that addition and multiplication can be done in unit time.

In this paper, we show that, for almost all numbers, AMCs and SLPs
need same asymptotic length for computation. However, for one specific
form of numbers, SLPs are strictly more powerful than AMCs by at least
one step of computation.
\end{abstract}

\section{Introduction}

\subsection{Straight Line Programs and $\mathrm{PosSLP}$}

Consider a sequence of registers $(a_{0},a_{1},\dots,a_{l})$ where
$a_{0}=1$ and $a_{l}=n$. A division-free Straight Line Program (SLP)
$P$ is a sequence of \emph{assignments }$\{(*_{i},j_{i},k_{i})\}_{i}$.
Each assignment $(*_{i},j_{i},k_{i})$ represents the operation $a_{i}\leftarrow a_{j_{i}}*_{i}a_{k_{i}}$,
where $*_{i}\in\{+,-,\times\}$ and $0\le j_{i},k_{i}<i$. $P$ can
be viewed as an arithmetic circuit computing integer $n$, where $1$
is the only given constant.

The number which $P$ computes or the result of $P$ is denoted by
$c(P)=n$. The length of $P$ is $l$, and is denoted by $|P|$. We
denote the length of the shortest SLP computing a number $n$ by $\tau(n)$.

The original motivation of this paper is this problem.
\begin{defn}
\textbf{PosSLP} is the following problem: Given a SLP $P$, does $P$
compute a positive number ?
\end{defn}
$\mathrm{PosSLP}$ was introduced by Allender et al. \cite{Allender:2006:CNA:1155440.1155565}
to study the relationship between classical models of computation
and computation over the reals in the Blum-Shub-Smale model \cite{CR98}. 

In the Blum-Shub-Smale model, the inputs are elements of $\bigcup_{n}\mathbb{R}^{n}\triangleq\mathbb{R}^{\infty}$,
and thus each machine accepts a {}``decision problem'' $L\subseteq\mathbb{R}^{\infty}$.
The machine can do all arithmetic operations and $\le$-test over
$\mathbb{R}$ in unit time, and it has a finite set $S$ of real \emph{machine
constants}. The set of decision problems accepted by polynomial-time
machine over $\mathbb{R}$ using only constants from $S\cup\{0,1\}$
is called $\P_{\mathbb{R}}^{S}$. The union of the classes $\P_{\mathbb{R}}^{S}$
over all $S$ is called \emph{polynomial-time over $\mathbb{R}$}
and is denoted $\P_{\mathbb{R}}$. The subclass $\P_{\mathbb{R}}^{\emptyset}$
of {}``constant-free polynomial-time'' is usually denoted by $\P_{\mathbb{R}}^{0}$. 

In order to relate computation over $\mathbb{R}$ to the classical
Boolean complexity classes such as $\P$, $\NP$, $\PSPACE$, etc.,
we consider the \emph{Boolean part} of decision problems over the
reals $L\subseteq\mathbb{R}^{\infty}$, which is defined as $\mathrm{BP}(L)=L\cap\{0,1\}^{*}$.
In other words, we consider only when the inputs are bit strings.
The Boolean part of $\P_{\mathbb{R}}$ is naturally denoted by $\mathrm{BP}(\P_{\mathbb{R}})=\{\mathrm{BP}(L)\mid L\in\P_{\mathbb{R}}\}$.
The relationship between $\mathrm{BP}(\P_{\mathbb{R}})$ and classical
classes has been shown: 
\[
\P/\mathsf{poly}\subseteq\mathrm{BP}(\P_{\mathbb{R}})\subseteq\PSPACE/\mathsf{poly}.
\]

The lower bound is by Koiran \cite{DBLP:journals/tcs/Koiran94} and
the upper bound is by Cucker and Grigoriev \cite{Cucker:1997:PRT:249388.249402}.
In the {}``constant-free'' case $\mathrm{BP}(\P_{\mathbb{R}}^{0})$,
it turns out that $\mathrm{PosSLP}$ precisely captures this class.
In particular, Allender et al. \cite{Allender:2006:CNA:1155440.1155565}
showed that

\[
\P^{\mathrm{PosSLP}}=\mathrm{BP}(\P_{\mathbb{R}}^{0}).
\]
This means that, given an oracle for $\mathrm{PosSLP}$, one can efficiently
simulate the constant-free%
\footnote{Actually, we have even that $\P^{\mathrm{PosSLP}}=\mathrm{BP}(\P_{\mathbb{R}}^{S})$
where $S$ contains only real algebraic constants, as $\P_{\mathbb{R}}^{S}=\P_{\mathbb{R}}^{0}$
\cite{Chapuis97saturationand}. %
} polynomial-time machines of the Blum-Shub-Smale model of real computation
when the inputs are bit strings. 

This shows interesting consequence because the well-known\emph{ Sum-of-square-roots
}problem, which has many applications to computational geometry, is
in $\mathrm{BP}(\P_{\mathbb{R}}^{0})$ \cite{Tiwari1992393}. In this
problem, we are given positive integers $d_{1},\dots,d_{n}$ and $k$,
and we want to decide whether $\sum_{i=1}^{n}\sqrt{d_{i}}\leq k$.
The \emph{Euclidean Traveling Salesman} problem, even though it has
$\PTAS$, is not known to be in $\NP$ because there is no known technique
for exactly verifying the length of a given path without using Sum-of-square-roots
as a procedure. This problem, hence, is in $\NP$ relative to the
\emph{Sum-of-square-roots} problem. 

The above discussion suggests the non-trivial upper bounds for $\mathrm{PosSLP}$
is very interesting. The current best upper bound for PosSLP is between
the $3^{rd}$ and the $4^{th}$ level of the counting hierarchy. In
particular, \textbf{$\mathrm{PosSLP}\in\P^{\PP^{\PP^{\PP}}}$}\cite{Allender:2006:CNA:1155440.1155565},
which implies that the Sum-of-square-roots problem and the Euclidean
Traveling Salesman Problem are also in the counting hierarchy.

There is yet another reason to be interested in $\mathrm{PosSLP}$.
In \cite{Allender:2006:CNA:1155440.1155565}, it was also shown that
the \emph{generic task of numerical computation} ($\mathrm{GTNC}$)
is polynomial time Turing equivalent to $\mathrm{PosSLP}$. Thus,
given the $\mathrm{PosSLP}$ oracle, we gain numerical stability for
any algorithms.

\subsection{Addition Multiplication Chains}

An Addition Multiplication Chain (AMC), introduced in \cite{Bahig2008611},
is an SLP without subtraction. That is, $*\in\{+,\times\}$. It can
be viewed as a monotone arithmetic circuit computing integer. We denote
the length of the shortest AMC computing a number $n$ by $\tau_{+}(n)$. 

Now, we ask whether for every number $n\in\mathbb{N}$, $\tau_{+}(n)\stackrel{?}{=}\tau(n)^{O(1)}$:
is the length of the shortest AMC computing $n$ always polynomial
in the length of the shortest SLP ? If the answer to this question
is {}``yes'', we can show that $\mathrm{PosSLP}$ is in the $2^{nd}$
level of the polynomial hierarchy. Note that $\PH\subseteq\P^{\PP}$
by Toda's theorem. So this is a much stronger upper bound.
\begin{claim}
If $\tau_{+}(n)=\tau(n)^{O(1)}$, then $\mathrm{PosSLP}\in\Sigma_{2}^{\P}$.\end{claim}
\begin{proof}
Given an SLP of size $\tau(n)$, we guess the AMC of size at most
$\tau(n)^{O(1)}$, and then verify the equivalence of the given SLP
and the guessed AMC. 

If the SLP computes a positive number, we can guess a correct AMC,
as $\tau_{+}(n)=\tau(n)^{O(1)}$ by assumption. If not, then no such
AMC exists because AMCs can compute only positive numbers.

To verify the equivalence, it is known that $\mathrm{EquSLP}$, deciding
if two SLPs compute the same number, is in $\mathrm{\coRP}$. So $\mathrm{PosSLP}\in\exists\coRP\subseteq\exists\coNP=\Sigma_{2}^{\P}$.
\end{proof}
On the other hand, it is also interesting to prove that the answer
is {}``no'': to show that there exists some number $n$ for which
$\tau_{+}(n)\notin\tau(n)^{O(1)}$. This would demonstrate how subtraction
helps generate integers super-polynomially faster, given that addition
and multiplication can be done in unit time.

\subsection{Our Results}

In this paper, we show two facts regarding the question, $\tau_{+}(n)\stackrel{?}{=}\tau(n)^{O(1)}$.
\begin{itemize}
\item First, for almost all numbers $n$, SLPs and AMCs need the same asymptotic
length for computing $n$.
\item Second, for one specific form of numbers, SLPs are indeed strictly
more powerful than AMCs by at least one step of computation. The hard
instance we use to establish the lower bound for AMCs is $2^{2^{n}}-1$,
where $n\in\mathbb{N}$.
\end{itemize}

\subsection{Polynomial Version of the Problem}

The question $\tau_{+}(n)\stackrel{?}{=}\tau(n)^{O(1)}$ can be rephrased
as {}``do monotone arithmetic circuits computing \emph{integer $n$}
have significantly less power than usual ones?''. 

In \emph{polynomial} setting, where we are given variables $x_{1},\dots,x_{n}$
and any field constants and we want to compute some polynomial $p(x_{1},\dots,x_{n})$,
this question is well-studied. There are polynomials that required
exponential size monotone arithmetic circuits, but can be computed
by polynomial size usual arithmetic circuits. 

For instance, Valiant \cite{Valiant:1979:NEP:800135.804412} showed
that such a hard instance is the perfect matching polynomial of the
triangle grid graph. One fact used to establish the lower bound for
monotone circuits is that every computed monomial can never be canceled.
To show the upper bound for usual circuits, the fact that perfect
matching polynomial is homogeneous is used; if we know that $p(x_{1},\dots,x_{n})$
is a $d$-degree homogeneous component of another easy polynomial
$q(x_{1},\dots,x_{n})$ that can be computed in $t$ steps, then $p$
can be computed in $O(d^{3}+d^{2}t)$ steps.

In our integer setting, we have none of these nice properties. It
seems that there is no integer counterpart for homogeneous polynomial.
Moreover, each bit of the computed integers can always be canceled
from one to zero. We refer the reader to a survey by Landau and Immerman
\cite{Landau94thesimilarities} for an analogy between polynomial
and integer problems.

\section{Tight asymptotic bound for almost all numbers}

In this section, we show that, for computing almost all numbers from
$2^{n}$ to $2^{n+1}$, both SLPs and AMCs of size $\Theta(n/\log n)$
are sufficient and necessary. For the proof, we use ideas by Brauer
\cite{Bra39} and Erdős \cite{Erd60}, who showed similar results
for addition chains. In \cite{Moreira97}, Moreira showed similar
non-asymptotic bounds for SLPs.
\begin{lem}
\label{lem:compute}For all numbers $z\in[2^{n},2^{n+1})$, $\tau(z),\tau_{+}(z)\in O(n/\log n)$.\end{lem}
\begin{proof}
Let $z\in[2^{n},2^{n+1})$ be any $(n+1)$-bit number, and let $k$
is an integer to be fixed later. We write $n+1=m\cdot k+r$ where
$0\le r<k$.

The idea is to compute $m$ $k$-bit numbers $u_{i}$, and then concatenate
them together to get $z$. We write 
\begin{eqnarray*}
z & = & u_{0}+2^{r}\sum_{j=1}^{m}u_{j}\cdot(2^{k})^{j-1}\\
 & = & u_{0}+2^{r}(u_{1}+2^{k}(u_{2}+2^{k}(\dots(u_{m-1}+2^{k}u_{m})\dots)
\end{eqnarray*}

Now we just compute all $1,2,3,\dots,2^{k}$, taking $2^{k}-1$ assignments.
From this, we have $2^{r}$, $2^{k}$ and $u_{i}<2^{k}$, for all
$i$. There are $m$ additions and $m$ multiplications in the expansion
above. 

Thus, $\tau_{+}(z)\le(2^{k}-1)+2m$. Choose $k=\left\lceil \log n-\log\log n\right\rceil $.
We get $\tau_{+}(z)\in O(n/\log n)$. 

Since SLPs are just AMCs equipped with subtraction, $\tau(z)\le\tau_{+}(z)$.
Thus, $\tau(z)\in O(n/\log n)$ as well.\end{proof}
\begin{lem}
\label{lem:count}AMCs and SLPs of length $n$ can represent at most
$2^{3n\log n}$ different numbers.\end{lem}
\begin{proof}
We will prove the result for only SLPs, because the proof for AMCs
is same. 

We count the possiblities for each assignment. For the $i^{th}$ assignment,
$a_{i}\leftarrow a_{j_{i}}*_{i}a_{k_{i}}$, there are at most $3\times i^{2}$
choices, as there are $3$ types of operation ($2$ types for AMCs)
and $j_{i},k_{i}<i$. Thus, there are at most $\prod_{i=1}^{n}3i^{2}=3^{n}(n!)^{2}\le2^{n\log3}\cdot2^{2n\log n}\le2^{3n\log n}$
combinations, for SLPs of length $n$. \end{proof}
\begin{thm}
For almost all numbers $z\in[2^{n},2^{n+1})$, $\tau(z),\tau_{+}(z)\in\Theta(n/\log n)$,
where {}``almost all'' means the fraction of $1-2^{-\epsilon n}$,
and $0<\epsilon<1$.\end{thm}
\begin{proof}
We only prove this for SLPs again. Let $d>1$ be some constant. By
Lemma \ref{lem:count}, SLPs of size at most $\frac{n}{3d\log n}$
can represent at most $2^{n/d}$ different numbers. 

Consider the numbers from $2^{n}$ to $2^{n+1}-1$. There are $2^{n}$
numbers but only $2^{n/d}$ of them can be computed by SLPs of size
$\frac{n}{3d\log n}$. Let $\epsilon=1-\frac{1}{d}$. Thus, there
are $2^{n}-2^{n/d}=2^{n}(1-2^{-\epsilon n})$ numbers which need SLPs
of size at least $\frac{n}{3d\log n}$. That is, almost all numbers
in the range required SLPs of size $\Omega(n/\log n)$.

But Lemma \ref{lem:compute} states that all numbers in this range
can be computed by SLPs of size $O(n/\log n)$. This concludes the
proof.
\end{proof}
Since any number $z$ is in $[2^{n},2^{n+1})$ for some $n$, we conclude
that $\tau_{+}(z)=\Theta(\tau(z))$ for almost all $z$.

\section{Gap of One}

In this section, we study the numbers of the form $2^{2^{n}}-1$,
where $n\in\mathbb{N}$. For convenience, let $N=2^{2^{n}}$. We show
that $\tau(N-1)\le n+2$, but $\tau_{+}(N-1)\ge n+3$.
\begin{claim}
\label{S_N-1 <=00003D n+2}$\tau(N-1)\le n+2$\end{claim}
\begin{proof}
We first compute $a_{1}\leftarrow a_{0}+a_{0}$. So $a_{1}=2$. Then
we square the current result $n$ times to get $a_{n+1}=2^{2^{n}}=N$.
Finally, $a_{n+2}\leftarrow a_{n+1}-a_{0}$, as desired.
\end{proof}
We prove the lower bound for $\tau_{+}(N-1)$ in 3 steps:
\begin{itemize}
\item There are more than $2$ additions in any AMC computing $N-1$.
\item There cannot be exactly $3$ additions in any AMC computing $N-1$.
\item There are less than $4$ additions in any AMC of length at most $n+2$
which computes $N-1$ .
\end{itemize}

\subsection{More than 2 Additions}

In this subsection, we show that there is no AMC computing $N-1$
using only $2$ or less additions. 

When there are few additions in an AMC, we know some facts about the
values of the registers $(a_{0},a_{1},\dots,a_{l})$ assigned in each
step.
\begin{fact}
\label{fact: factor-by-addition }The registers assigned in an AMC
$P$ between the $k^{th}$ and $k+1^{th}$ additions have values of
the form $\prod_{i=1}^{k}\alpha_{i}^{e_{i}}$ where $\alpha_{i}=\prod_{j<i}\alpha_{j}^{f_{j}}+\prod_{j<i}\alpha_{j}^{f'_{j}}$
and, for all $j$, $\min\{f_{j},f'_{j}\}=0$. In particular, $\alpha_{1}=2$,
and $\alpha_{2}=2^{c}+1$, for some $c\ge0$.\end{fact}
\begin{proof}
We prove the statement by induction on $k$. For the base case, the
numbers computed in $P$ before the first addition can be only $1=\prod_{i=1}^{0}\alpha_{i}^{e_{i}}$. 

For the inductive step, consider the $k^{th}$ addition assignment
$a\leftarrow b+c$, where $b=\prod_{i=1}^{k-1}\alpha_{i}^{b_{i}}$
and $c=\prod_{i=1}^{k-1}\alpha_{i}^{c_{i}}$, by the induction hypothesis.
Thus, $a=\prod_{i=1}^{k-1}\alpha_{i}^{b_{i}}+\prod_{i=1}^{k-1}\alpha_{i}^{c_{i}}$.
We pull the shared factors out and get $a=\prod_{i=1}^{k-1}\alpha_{i}^{\min(b_{i},c_{i})}\times\alpha_{k}$
where $\alpha_{k}=\prod_{i=1}^{k-1}\alpha_{i}^{f_{i}}+\prod_{i=1}^{k-1}\alpha_{i}^{f'_{i}}$
and, for all $i$, $\min\{f_{i},f'_{i}\}=0$. Since multiplications
can not yield new factors into the sequence, all the numbers computed
before the $k+1^{th}$ addition would be of the form $\prod_{i=1}^{k}\alpha_{i}^{e_{i}}$. 

The first factor is $\alpha_{1}=2$, because the first addition is
$1+1$. Since we just show that $\alpha_{2}=\alpha_{1}^{f_{1}}+\alpha_{1}^{f'_{1}}$
where $\min\{f_{1},f'_{1}\}=0$, as $\alpha_{1}=2$, it follows that
$\alpha_{2}=2^{c}+1$, for some $c\ge0$.
\end{proof}
Next, we state two known properties of the number $N-1$ and Fermat
number (cf. \cite{krizek200217}).
\begin{fact}
\label{fact: factor of N-1}$N-1=F_{n}-2=\prod_{i=0}^{n-1}F_{i}$
where $F_{i}=2^{2^{i}}+1$ is the $i^{th}$ Fermat number.
\end{fact}

\begin{fact}
\label{fact:Fermat-coprime}Fermat numbers are coprime to each other.
\end{fact}
\begin{lem}
\label{lem:low expo}In an AMC with $k$ additions computing $N-1=\prod_{i=1}^{k}\alpha_{i}^{e_{i}}$,
we have $e_{1}=0$, and $e_{2}\le1$. Moreover, if $e_{2}=1$, then
there is a unique index $j$ such that $F_{j}\mid\alpha_{2}$.\end{lem}
\begin{proof}
By Fact \ref{fact: factor-by-addition }, $\alpha_{1}=2$. But $N-1$
is odd, so $e_{1}=0$. 

As we know $\alpha_{2}=2^{c}+1$, for some $c$, we now claim that
there is a unique index $j$, such that $F_{j}\mid\alpha_{2}$. 

There are $2$ cases. If $c$ is a power of $2$, $c=2^{j}$, then
$\alpha_{2}=2^{2^{j}}+1=F_{j}$. Hence, $j$ is a unique index such
that $F_{j}\mid\alpha_{2}$ by Fact \ref{fact:Fermat-coprime}.

Otherwise, $c$ is not a power of $2$, $c=2^{j}\cdot w$, for some
odd $w$, then $\alpha_{2}=(2^{2^{j}})^{w}+1=(-1)^{w}+1=0\bmod F_{j}$.
Thus, $F_{j}\mid\alpha_{2}$. For uniqueness, we show that for $k\neq j$,
$F_{k}\nmid\alpha_{2}$. If $k<j$, then $\alpha_{2}=(2^{2^{k}})^{2^{j-k}\cdot w}+1=(-1)^{2^{j-k}\cdot w}+1=1+1=2\bmod F_{k}$.
If $k>j$, then we write $w=w'\cdot2^{k-j}+w''$, where $w''<2^{k-j}$.
We have $\alpha_{2}=2^{2^{j}w}+1=2^{2^{k}w'}\cdot2^{2^{j}w''}+1=(\pm1)\cdot2^{2^{j}w''}+1\bmod F_{k}$.
Since $w''<2^{k-j}$, so $2^{2^{j}w''}+1<2^{2^{k}}+1=F_{k}$. So $\alpha_{2}=\pm2^{2^{j}w''}+1\neq0\bmod F_{k}$.

Now, we show that $e_{2}\le1$. Since, $F_{j}^{e_{2}}\mid\alpha_{2}^{e_{2}}\mid N-1=\prod_{i=0}^{n-1}F_{i}$.
If $e_{2}>1$, then $F_{j}\mid\prod_{i\neq j}F_{i}$ . But this contradicts
the coprimality of Fermat numbers, see Fact \ref{fact:Fermat-coprime}.\end{proof}
\begin{thm}
\label{thm:more than 2}If AMC computes $N-1$, there are strictly
more than $2$ additions.\end{thm}
\begin{proof}
If there is no addition, then the result is trivially $1$. 

If there is $1$ addition, then the result can only be even number,
but $N-1$ is odd.

If there are $2$ additions, $N-1=\alpha_{1}^{e_{1}}\alpha_{2}^{e_{2}}$.
By Lemma \ref{lem:low expo}, it must be that $N-1=\alpha_{2}$, and
there is only one $j$, where $F_{j}\mid\alpha_{2}$. So for $k\neq j$,
$F_{k}\nmid N-1$, which contradicts Fact \ref{fact: factor of N-1}
: $N-1=\prod_{i=0}^{n-1}F_{i}$.
\end{proof}

\subsection{Not 3 Additions}

To prove the main statement in this subsection, we need this lemma. 
\begin{lem}
\label{lem:no 1 expo}Let $S=\{2^{2^{0}},2^{2^{1}},2^{2^{2}},\dots\}$
be the sequence of all numbers of the form $2^{2^{i}}$. Let $A\subset\mathbb{Z}\backslash\{0,\pm1\}$,
$B\subset\mathbb{Z}^{+}\backslash\{1\}$ and \textup{$C\subset\mathbb{Z}^{+}$}
be finite sets.

For any increasing sequence $T$ of the form \textup{$\{c_{i}b_{i}^{e_{i}}+a_{i}\}_{i}$,
where $e_{i}\in\mathbb{Z}^{+}$ and $a_{i},b_{i},c_{i}\in A,B,C$
respectively, there are infinitely many indices $i$ such that $S_{i}\neq T_{i}$.
}%
\footnote{$a_{i}$ does not denote the $i^{th}$ element of $A$, but it is
just some element in $A$ which constitutes the $i^{th}$ element
of $T$, $c_{i}b_{i}^{e_{i}}+a_{i}$. $b_{i}$ and $c_{i}$ have similar
meaning. %
}\end{lem}
\begin{proof}
Let $\bar{a},\bar{c}$ be the maximum elements in $A$ and $C$ respectively.

First, since $A$ is finite and $|a_{i}|>1$, so there exists a number
$t\in\mathbb{N}$ such that for all $a_{i}$, $a_{i}^{2^{t}}>\bar{a}$.
We fix such $t$. 

Again, since $A,B$ and $C$ are finite and $T$ is increasing , it
follows that $\{e_{i}\}_{i}$ is unbounded. So there exists $i$ such
that $b_{i}^{e_{i}}>\bar{c},\bar{a}^{2^{t\cdot|B|+1}}$. Fix this
$i$.

Consider the sequence $\bar{B}=\{b_{i},b_{i+t},\dots,b_{i+t\cdot|B|}\}$,
there are $b_{j},b_{j'}\in\bar{B}$ such that $b_{j}=b_{j'}$ where
$i\le j<j'$, because the size of $\bar{B}$ is more than $|B|$.
Note that $t\le j'-j\le t\cdot|B|$.

Now, we claim that for at least one index $k\in\{i,j,j'\}$, we have
$S_{k}\neq T_{k}$. Suppose not, then we have $c_{k}b_{k}^{e_{k}}+a_{k}=2^{2^{k}}$
for all $k\in\{i,j,j'\}$. Thus, we get
\begin{eqnarray}
(c_{i}b_{i}^{e_{i}}+a_{i})^{2^{j-i}} & = & c_{j}b_{j}^{e_{j}}+a_{j}\label{eq:1}\\
\mbox{and }(c_{j}b_{j}^{e_{j}}+a_{j})^{2^{j'-j}} & = & c_{j'}b_{j'}^{e_{j'}}+a_{j'}\label{eq:2}
\end{eqnarray}

We will show that $b_{j'}^{e_{j'}}\ge b_{j}^{e_{j}}\ge b_{i}^{e_{i}}$.
To prove $b_{j}^{e_{j}}\ge b_{i}^{e_{i}}$, if $i=j$, then it is
trivial. If $i<j$, suppose that $b_{j}^{e_{j}}<b_{i}^{e_{i}}$. This
contradicts Equation (\ref{eq:1}), since 
\begin{eqnarray*}
(c_{i}b_{i}^{e_{i}}+a_{i})^{2^{j-i}} & \ge & (c_{i}b_{i}^{e_{i}}+a_{i})^{2}\\
 & = & c_{i}^{2}b_{i}^{2e_{i}}+2a_{i}c_{i}b_{i}^{e_{i}}+a_{i}^{2}\\
 & > & b_{i}^{2e_{i}}+b_{i}^{e_{i}}\\
 & > & c_{j}b_{j}^{e_{j}}+a_{j}
\end{eqnarray*}
The last inequality follows because $\ensuremath{b_{i}^{e_{i}}>b_{j}^{e_{j}},\bar{c},\bar{a}}$.

We can prove $b_{j'}^{e_{j'}}\ge b_{j}^{e_{j}}$ in exactly same way
using Equation (\ref{eq:2}). Now, since $b_{j}=b_{j'}$, we have
$e_{j'}\ge e_{j}$. Consider Equation (\ref{eq:2}) modulo $b_{j}^{e_{j}}$.
We have

\[
a_{j}^{2^{j'-j}}\equiv a_{j'}\bmod b_{j}^{e_{j}}
\]

It follows that $a_{j}^{2^{j'-j}}=a_{j'}$ because $b_{j}^{e_{j}}\ge b_{i}^{e_{i}}>\bar{a}^{2^{t\cdot|B|+1}}>a_{j}^{2^{j'-j}},a_{j'}$.
But $t$ is such that $a_{j}^{2^{j'-j}}\ge a_{j}^{2^{t}}>\bar{a}\ge a_{j'}$,
so we reach a contradiction. Thus, we have shown that for all $i$
such that $b_{i}^{e_{i}}>\bar{c},\bar{a}^{2^{t\cdot|B|+1}}$, we can
find at least one index $k\in\{i,j,j'\}$ where $S_{k}\neq T_{k}$,
by considering the finite sequence $\bar{B}$. Now, since there are
infinitely many $i$ satisfying $b_{i}^{e_{i}}>\bar{c},\bar{a}^{2^{t\cdot|B|+1}}$,
we are done.\end{proof}
\begin{lem}
\label{lem: N-1 w/ 3+}If an AMC computes $N-1$ using exactly $3$
additions, then $N-1=2^{d_{1}}\alpha_{2}^{d_{2}}+\alpha_{2}^{d_{2}'}$
where $d_{2}\le1$ or $d_{2}'\le1$.\end{lem}
\begin{proof}
By Lemma \ref{lem:low expo}, $N-1=\alpha_{2}^{e_{2}}\alpha_{3}^{e_{3}}$
where $e_{2}\le1$. We claim that $e_{3}=1$. 

There is at most one $j$ for which $F_{j}\mid\alpha_{2}$. So, for
$k\neq j$, it must be that $F_{k}\mid\alpha_{3}$. We have $F_{k}^{e_{3}}\mid\alpha_{3}^{e_{3}}\mid N-1=\prod_{i=0}^{n-1}F_{i}$.
Now, if $e_{3}>1$, then $F_{k}\mid\prod_{i\neq k}F_{i}$, which contradicts
the coprimality of Fermat numbers, see Fact \ref{fact:Fermat-coprime}. 

We also have $e_{3}\neq0$, otherwise it must be that $N-1=\alpha_{2}$
which is impossible as shown in Theorem \ref{thm:more than 2}.

By Fact \ref{fact: factor-by-addition }, $\alpha_{3}=\alpha_{1}^{f_{1}}\alpha_{2}^{f_{2}}+\alpha_{1}^{f'_{1}}\alpha_{2}^{f'_{2}}$,
where $\min\{f_{i},f'_{i}\}=0$. So $N-1=\alpha_{2}^{e_{2}}\alpha_{3}=\alpha_{1}^{f_{1}}\alpha_{2}^{e_{2}+f_{2}}+\alpha_{1}^{f'_{1}}\alpha_{2}^{e_{2}+f'_{2}}$.
Since, $\min\{e_{2}+f_{2},e_{2}+f'_{2}\}\le1$, we have the claim
by rewriting the variables.\end{proof}
\begin{thm}
\label{thm:not 3}There is no AMC computing $N-1$ using exactly $3$
additions.\end{thm}
\begin{proof}
By Fact \ref{fact: factor-by-addition } and Lemma \ref{lem: N-1 w/ 3+},
$N-1=2^{d_{1}}\alpha_{2}^{d_{2}}+\alpha_{2}^{d_{2}'}$ where $\alpha_{2}=2^{c}+1$.
We will show that, for any value of $c,d_{1},d_{2}$ and $d_{2}'$,
an AMC with $3$ additions cannot compute $N-1$.

First, neither $c$ nor $d_{1}$ can be zero. If $d_{1}=0$, then
$2^{0}\alpha_{2}^{d_{2}}+\alpha_{2}^{d_{2}'}$ is even while $N-1$
is odd. If $c=0$, then $N-1=2^{d_{1}+d_{2}}+2^{d_{2}'}=2^{k}+1$
for some $k$, as $N-1$ is odd. As in Lemma \ref{lem:low expo},
if $N-1$ is in this form, then there is only one index $j$ such
that $F_{j}\mid N-1$, and this contradicts the fact that $N-1=\prod_{i=0}^{n-1}F_{i}$.

Now, there are 4 cases which all will lead to a contradiction : when
$c,d_{1}\ge2$, when $c\ge3,d_{1}=1$, when $c=1,d_{1}\ge3$, and
some remaining small cases, namely, $(c,d_{1})=(1,1),(2,1)$ or $(1,2)$. 

We use this following equality in the first two cases, 
\[
\alpha_{2}^{d}-1=(\alpha_{2}-1)\sum_{i=0}^{d-1}\alpha_{2}=2^{c}\sum_{i=0}^{d-1}\alpha_{2}.
\]

\begin{caseenv}
\item $c,d_{1}\ge2$.

We have

\begin{eqnarray*}
N-1 & = & 2^{d_{1}}\alpha_{2}^{d_{2}}+\alpha_{2}^{d_{2}'}\\
N & = & 2^{d_{1}}\alpha_{2}^{d_{2}}+(\alpha_{2}^{d_{2}'}-1)+2\\
N & = & 2^{d_{1}}\alpha_{2}^{d_{2}}+2^{c}\sum_{i=0}^{d_{2}'-1}\alpha_{2}+2.
\end{eqnarray*}
Since, $N=2^{2^{n}}$ and $c,d_{1}\ge2$, the LHS is divisible by
$4$, but the RHS is not.

\item $c\ge3\mbox{ and }d_{1}=1$.

We continue manipulating the equation from the last case, and set
$d_{1}=1$.

\begin{eqnarray*}
N & = & 2\alpha_{2}^{d_{2}}+2^{c}\sum_{i=0}^{d_{2}'-1}\alpha_{2}+2\\
N/2 & = & \alpha_{2}^{d_{2}}+2^{c-1}\sum_{i=0}^{d_{2}'-1}\alpha_{2}+1\\
N/2 & = & (\alpha_{2}^{d_{2}}-1)+2^{c-1}\sum_{i=0}^{d_{2}'-1}\alpha_{2}+2\\
N/2 & = & 2^{c}\sum_{i=0}^{d_{2}-1}\alpha_{2}+2^{c-1}\sum_{i=0}^{d_{2}'-1}\alpha_{2}+2.
\end{eqnarray*}
Since, $c\ge3$, the LHS is divisible by $4$, but the RHS is not.

\item $c=1\mbox{ and }d_{1}\ge3$.

We have $\alpha_{2}=3$.

\begin{eqnarray*}
N-1 & = & 2^{d_{1}}3{}^{d_{2}}+3^{d_{2}'}\\
N & = & 2^{d_{1}}3{}^{d_{2}}+(3^{d_{2}'}+1).
\end{eqnarray*}

One can check that, for any $e$, $3^{e}+1$ is congruent to only
$2$ or $4$ modulo $8$. So the LHS is divisible by $8$, but the
RHS is not.

\item The remaining cases.

This last case actually consists of many small cases. These are the
cases when $(c,d_{1})=(1,1),(2,1)$ or $(1,2)$, and we also have
that $d_{2}\le1$ or $d_{2}'\le1$ by Lemma \ref{lem: N-1 w/ 3+}.

Fortunately, they can be handled simultaneously using Lemma \ref{lem:no 1 expo}.
This lemma tells us that, if the computed numbers are of the form
$\{c_{i}b_{i}^{e_{i}}+a_{i}\}_{i}$, where $a_{i},b_{i},c_{i}$ are
restricted to be from finite sets, and only $e_{i}$ can be any number,
then there are infinitely many $n$, where $2^{2^{n}}$ does not occur
in the sequence $\{c_{i}b_{i}^{e_{i}}+a_{i}\}_{i}$.

In this last case, we will show that $2^{d_{1}}\alpha_{2}^{d_{2}}+\alpha_{2}^{d_{2}'}+1$
is restricted to be of the form $\{c_{i}b_{i}^{e_{i}}+a_{i}\}_{i}$
as in the lemma. So there are infinitely many $n$, where $2^{2^{n}}=N\neq2^{d_{1}}\alpha_{2}^{d_{2}}+\alpha_{2}^{d_{2}'}+1$,
which would finish the proof.

To see this, observe that either $d_{2}$ or $d_{2}'$ is at most
$1$, hence only one of $d_{2}$ or $d_{2}'$ can be free variable.
For example, when $d_{2}'$ is not fixed and $(c,d_{1},d_{2})=(1,2,1)$,
we have $2^{d_{1}}\alpha_{2}^{d_{2}}+\alpha_{2}^{d_{2}'}+1=3^{d_{2}'}+13=N$.
In this case, $c_{i}=1$, $b_{i}=3$, and $a_{i}=13$. After we list
all possible cases, we have $B=\{3,5\}$, $C=\{1,2,4\}$, $A=\{2,3,4,5,6,7,11,13\}$,
which are all finite sets and comply with the conditions of the lemma.

\end{caseenv}
\end{proof}

\subsection{Less than 4 Additions}

In this subsection, we show that, if the length of an AMC $P$ is
at most $n+2$, then $P$ cannot compute the large number $N-1$ using
$4$ or more additions.

The main idea is that, to compute a large number using a short AMC,
if there are too many additions, then there are not enough squaring
steps. Indeed, to square the largest number is the fastest way to
get a new big number. Hence, the final result cannot be as large as
the desired number.

We show the following facts to argue this formally.

We say that two AMCs $P=\{(*_{i},j_{i},k_{i})\}_{i}$ and $P'=\{(*'_{i},j'_{i},k'_{i})\}_{i}$
are \emph{equal (assignment-wise)} if $(*_{i},j_{i},k_{i})=(*'_{i},j'_{i},k'_{i})$
for all $i$. We also say $P$ is \emph{transformed} into $P'$ or
$P'$ is obtained from $P$, if $P$ is manipulated such that it is
equal to $P'$.
\begin{fact}
Let $P$ and $P'$ be any AMCs with same numbers of additions and
multiplications. We can transform $P$ into $P'$ by some sequence
of these two operations:
\begin{enumerate}
\item Swapping the type of the operations of the $i^{th}$ assignment with
$(i+1)^{th}$ assignment. That is, $*_{i}\leftarrow*_{i+1}$ and $*_{i+1}\leftarrow*_{i}$.
\item Incrementing/decrementing the index $j_{i}$ or $k_{i}$ of the $i^{th}$
assignment. That is, $j_{i}\leftarrow j_{i}\pm1$ or $k_{i}\leftarrow k_{i}\pm1$.
\end{enumerate}
\end{fact}
\begin{defn}
\textbf{Irredundant AMC}. An AMC $P$ is irredundant iff, for every
$i\in[0,|P|-1]$, $a_{i}<a_{i+1}$ and $a_{i}$ is used as the operand
in some assignment. 
\end{defn}
Here are easy facts about irredundant AMCs.
\begin{fact}
\label{fact:short irred}If an AMC $Q$ is redundant, then there exists
an irredundant AMC $P$, where $c(P)=c(Q)$ but $|P|<|Q|$. 
\end{fact}

\begin{fact}
\label{fact:inc/dec irred}For any irredundant AMC $P$, any increase
(decrease) of the value of the register $a_{i}$ of $P$ always increases
(decreases) the computed number $c(P)$.
\end{fact}
In the followings, we consider only irredundant AMCs because of this
fact.
\begin{fact}
If there is no irredundant AMC of size at most $n+2$ which computes
$N-1$ using $4$ or more additions, then neither is the redundant
one.\end{fact}
\begin{proof}
We prove the contra-positive. Suppose that $Q$ is a redundant AMC
of size at most $n+2$ which computes $N-1$ using $4$ or more additions.
Then there is an AMC $P$, where $c(P)=c(Q)=N-1$ but $|P|<|Q|\le n+2$
by Fact \ref{fact:short irred}. There are also at least $4$ additions
in $P$, otherwise this contradicts Theorem \ref{thm:more than 2}
or \ref{thm:not 3}.\end{proof}
\begin{fact}
\label{fact:inc->more}Let $P$ be any irredundant AMC, and let $\tilde{P}$
be obtained from $P$ by incrementing some index $j_{i}$ or $k_{i}$.
We have $c(P)<c(\tilde{P})$.\end{fact}
\begin{proof}
As $P$ is irredundant, $a_{j_{i}}<a_{j_{i}+1}$ and $a_{k_{i}}<a_{k_{i}+1}$.
Hence, $a_{i}$ would be increased. Then this fact follows immediately
from Fact \ref{fact:inc/dec irred}.
\end{proof}
We say that the AMC $P$ has \emph{maximum indices} if $j_{i}=k_{i}=i-1$,
for all $i$. That is, every register depends on only the previous
one. Note that AMCs with maximum indices are irredundant. 
\begin{fact}
\label{fact:swap->more}Let $P$ be any AMC with maximum indices,
and let $\tilde{P}$ be obtained from $P$ by swapping the operation
$*_{i}=\times$ with $*_{i+1}=+$. We have $c(P)<c(\tilde{P})$.\end{fact}
\begin{proof}
For every $i$, let $\tilde{a}_{i}$ be the $i^{th}$ register of
$\tilde{P}$. Since $P$ and $\tilde{P}$ have maximum indices, every
register depends on only the previous one. It suffices to show that
$a_{i+1}<\tilde{a}_{i+1}$.

We have $a_{i+1}=2a_{i}=2a_{i-1}^{2}$ and $\tilde{a}_{i+1}=\tilde{a}_{i}^{2}=(2\tilde{a}_{i-1})^{2}=(2a_{i-1})^{2}$.
Hence, $a_{i+1}<\tilde{a}_{i+1}$ as desired.\end{proof}
\begin{lem}
\label{lem:max of a,m}The maximum number that irredundant AMCs with
$a$ additions and $m$ multiplications can compute is $2^{a\cdot2^{m}}$. \end{lem}
\begin{proof}
Consider an AMC $P^{*}$ with maximum indices and the first $a$ assignments
are additions, followed by $m$ multiplications. Actually, these additions
are doubling and multiplications are squaring as $P^{*}$ has maximum
indices. Thus, $c(P)=(2^{a})^{2^{m}}$.

We claim that, for any other irredundant AMC $P$ with $a$ additions
and $m$ multiplications, $c(P)<c(P^{*})$. If $P$ does not have
not maximum indices, then $c(P)<c(P^{*})$ by Fact \ref{fact:inc->more}.
If $P$ has maximum indices but $P$ is not equal to $P^{*}$, then
there exists an index $i$ where $(*_{i},*_{i+1})=(\times,+)$. Hence,
$c(P)<c(P^{*})$ by Fact \ref{fact:swap->more}. So $P^{*}$ is the
irredundant AMC with $a$ additions and $m$ multiplications which
computes the biggest possible number.\end{proof}
\begin{lem}
The second largest number that irredundant AMCs with $a$ additions
and $m$ multiplications can compute is $2^{e}$ where $e=\max\{\log3\cdot(a-2)\cdot2^{m},3a\cdot2^{m-2},(2a-1)\cdot2^{m-1}\}$.\end{lem}
\begin{proof}
Consider $P^{*}$ defined in the last lemma. We claim that, any irredundant
AMC $P$ computing the second largest number can be transformed into
$P^{*}$ in one manipulation.

Suppose that we need at least $2$ manipulations to transform $P$
into $P^{*}$. Since $P^{*}$ has maximum indices, the manipulation
contains no index decrement. Since two kinds of manipulation do not
depend on each other, we assume that we always increment the indices
and then swap the types of operations.

There are 2 cases. First, we need to increment index. After the first
increment, we have $\tilde{P}$ which is not $P^{*}$ yet. We have
$c(P)<c(\tilde{P})<c(P^{*})$ by Fact \ref{fact:inc->more} and Lemma
\ref{lem:max of a,m}. Second, we don't need to increment index, so
$P$ has maximum indices. After we swap the types of operations, we
get $\tilde{P}$ which is not $P^{*}$ yet. We again have $c(P)<c(\tilde{P})<c(P^{*})$
by Fact \ref{fact:swap->more} and Lemma \ref{lem:max of a,m}. In
both cases, $P$ computes at most the third largest number. 

Now, having that we can get $P^{*}$ from $P$ in one manipulation,
we can get $P$ from $P^{*}$ in one manipulation as well. There are
many ways to manipulate $P^{*}$ to get $P$ in one step. But it turns
out that there are only $3$ possible values of $c(P)$.

First, we decrement one index of the addition assignment in $P^{*}$.
So there is the assignment $a_{i}\leftarrow a_{i-1}+a_{i-2}$ in $P$.
The value of the register at the last addition is $2^{a-i}(2^{i-1}+2^{i-2})=2^{a-2}\cdot3$.
Note that this does not depend on $i$. Then we square for $m$ times,
and get $c(P)=(2^{a-2}\cdot3)^{2^{m}}=2^{\log3\cdot(a-2)\cdot2^{m}}$.

Second, we decrement one index of the multiplication assignment in
$P^{*}$. Hence, there is the assignment $a_{a+i}\leftarrow a_{a+i-1}\times a_{a+i-2}$
in $P$. Similarly, we have $c(P)=(2^{a\cdot2^{i-1}}\times2^{a\cdot2^{i-2}})^{2^{m-i}}=2^{a\cdot3\cdot2^{m-2}}$. 

Third, we swap the last addition with the first multiplication in
$P^{*}$. We have $c(P)=((2^{a-1})^{2}\cdot2)^{2^{m-1}}=2^{(2a-1)\cdot2^{m-1}}$.

The second largest number is the maximum among these results.
\end{proof}
Now, we are ready for the main statement. 
\begin{thm}
If AMC computing $N-1$ has size at most $n+2$, there are less than
$4$ additions.\end{thm}
\begin{proof}
Let $a$ be number of additions. If $a\ge5$, then the biggest possible
number is $2^{a\cdot2^{m}}\le2^{2^{(n+2-a)+\log a}}\le2^{2^{n-3+\log5}}=2^{2^{n-0.67}}<N-1$,
for large enough $n$. 

If $a=4,$ then the maximum number is $2^{2^{(n+2-4)+\log4}}=N>N-1$.
Also, the second largest number is $2^{e}$ where $e=\max\{\log3(4-2)\cdot2^{n-2},3\cdot4\cdot2^{n-4},(2\cdot4-1)\cdot2^{n-3}\}=(2\cdot4-1)\cdot2^{n-3}=2^{n-3+\log7}$.
We can see that $2^{2^{n-3+\log7}}<N-1$, for large enough $n$.
\end{proof}
We have shown that any AMC $P$ computing $N-1$ must use strictly
more than $3$ additions by Theorem \ref{thm:more than 2} and \ref{thm:not 3},
but if $P$ has size at most $n+2$, then must be less than $4$ additions
as well. So $\tau_{+}(N-1)\ge n+3$. As $\tau(N-1)\le n+2$ by Claim
\ref{S_N-1 <=00003D n+2}, we conclude that subtraction helps us compute
$N-1$ strictly faster by at least one step.

\section{Conclusion}

We have shown that, for almost all numbers $n$, $\tau_{+}(n)=\Theta(\tau(n))$.
We still suspect that there are hard instances $n$ showing that SLPs
compute some numbers super-polynomially faster than AMCs, $\tau_{+}(n)\notin\tau(n)^{O(1)}$.
In this paper, we demonstrate the power of subtraction by showing
only that $\tau_{+}(N-1)-\tau(N-1)\ge1$, where $N=2^{2^{n}}$ and
$n\in\mathbb{N}$. 

However, note that $N-1$ is not such hard instance because $\tau_{+}(N-1)\le2n$.
This can be shown using the facts that $N-1=F_{n}-2$ and, for all
$i$, $F_{i}-2=F_{i-1}\times(F_{i-1}-2)$. 

It remains as an open problem to decide whether there exists $n$
such that $\tau_{+}(n)\notin\tau(n)^{O(1)}$. The next step towards
this problem is to find $n$ such that $\tau_{+}(n)-\tau(n)=f(n)$
where $f$ is some increasing function.

\section*{Acknowledgments}

We thank Markus Bläser for introducing us the problem and giving us
useful comments and suggestions. Without his guidance, this work could
not have been possible.


\begin{thebibliography}{10}

\bibitem{Allender:2006:CNA:1155440.1155565}
Eric Allender, Peter B\"urgisser, Johan Kjeldgaard-Pedersen, and Peter~Bro
  Miltersen.
\newblock On the complexity of numerical analysis.
\newblock In {\em Proceedings of the 21st Annual IEEE Conference on
  Computational Complexity}, CCC '06, pages 331--339, Washington, DC, USA,
  2006. IEEE Computer Society.

\bibitem{Bahig2008611}
Hatem~M. Bahig.
\newblock On a generalization of addition chains: Addition-multiplication
  chains.
\newblock {\em Discrete Mathematics}, 308(4):611 -- 616, 2008.

\bibitem{Bra39}
Alfred Brauer.
\newblock On addition chains.
\newblock {\em Bull. Amer. Math. Soc.}, 45(10):736--739, 1939.

\bibitem{Chapuis97saturationand}
Olivier Chapuis and Pascal Koiran.
\newblock Saturation and stability in the theory of computation over the reals,
  1997.

\bibitem{Cucker:1997:PRT:249388.249402}
Felipe Cucker and Dima Grigoriev.
\newblock On the power of real turing machines over binary inputs.
\newblock {\em SIAM J. Comput.}, 26(1):243--254, February 1997.

\bibitem{Moreira97}
Carlos Gustavo~T. de~A.~Moreira.
\newblock On asymptotic estimates for arithmetic cost functions.
\newblock In {\em Proceedings of the American Mathematical Society}, volume
  125, pages 347--353, 1997.

\bibitem{Erd60}
Paul Erd\H{o}s.
\newblock Remarks on number theory. iii: On addition chains (in english).
\newblock {\em Acta Arith.}, 6:77--81, 1960.

\bibitem{DBLP:journals/tcs/Koiran94}
Pascal Koiran.
\newblock Computing over the reals with addition and order.
\newblock {\em Theor. Comput. Sci.}, 133(1):35--47, 1994.

\bibitem{krizek200217}
M.~Krizek, F.~Luca, and L.~Somer.
\newblock {\em 17 Lectures on Fermat Numbers: From Number Theory to Geometry}.
\newblock Springer, 2002.

\bibitem{CR98}
M.~Shub L.~Blum, F.~Cucker and S.~Smale.
\newblock {\em Complexity and Real Computation}.
\newblock Springer, 1998.

\bibitem{Landau94thesimilarities}
Susan Landau and Neil Immerman.
\newblock The similarities (and differences) between polynomials and integers.
\newblock Technical report, 1994.

\bibitem{Tiwari1992393}
Prasoon Tiwari.
\newblock A problem that is easier to solve on the unit-cost algebraic ram.
\newblock {\em Journal of Complexity}, 8(4):393 -- 397, 1992.

\bibitem{Valiant:1979:NEP:800135.804412}
L.~G. Valiant.
\newblock Negation can be exponentially powerful.
\newblock In {\em Proceedings of the eleventh annual ACM symposium on Theory of
  computing}, STOC '79, pages 189--196, New York, NY, USA, 1979. ACM.

\end{thebibliography}
\end{document}